\documentclass[conference, twocolumn]{IEEEtran}
\IEEEoverridecommandlockouts
\usepackage{cite}
\usepackage{amsmath,amssymb,amsfonts}
\usepackage{dsfont}
\usepackage{algorithmic}
\usepackage{graphicx}
\usepackage{textcomp}
\usepackage{xcolor}
\usepackage{enumitem}

\usepackage{mathtools}

\usepackage{amsthm}
\usepackage{arydshln}
\usepackage{multirow}
\usepackage{calc}
\usepackage{blkarray}
\usepackage{url}
\theoremstyle{definition}

\newtheorem{lemma}{Lemma}
\newtheorem{claim}{Claim}

\newtheorem{example}{Example}
\newtheorem{remark}{Remark}
\newtheorem{definition}{Definition}
\usepackage{graphicx,cite}
\usepackage{dblfloatfix}
\usepackage[justification=centering]{caption}
\usepackage{blindtext, graphicx, amsfonts,
	amssymb,multirow,epstopdf}
\usepackage[linesnumbered,ruled]{algorithm2e}
\def\BibTeX{{\rm B\kern-.05em{\sc i\kern-.025em b}\kern-.08em
    T\kern-.1667em\lower.7ex\hbox{E}\kern-.125emX}}

\newcommand{\frakC}{\mathfrak{C}}

\newcommand{\bfA}{\mathbf{A}}
\newcommand{\bfB}{\mathbf{B}}

\newcommand{\bfG}{\mathbf{G}}

\newcommand{\bfx}{\mathbf{x}}

\newcommand{\bfv}{\mathbf{v}}
\newcommand{\bbF}{\mathbb{F}}
\newcommand{\bfy}{\mathbf{y}}

\begin{document}

\title{Universally Decodable Matrices for Distributed Matrix-Vector Multiplication}
\author{\IEEEauthorblockN{Aditya Ramamoorthy, Li Tang}
	\IEEEauthorblockA{Department of Electrical and Computer Engineering \\
		Iowa State University\\
		Ames, IA 50010, U.S.A. \\
	}\and
	\IEEEauthorblockN{Pascal O. Vontobel}
	\IEEEauthorblockA{Department of Information Engineering \\
		The Chinese University of Hong Kong\\
		Hong Kong, S. A. R. \\
	}
	\thanks{This work was supported in part by the National Science Foundation (NSF) under grant CCF-1718470.}
}

\maketitle

%

\begin{abstract}
	Coded computation is an emerging research area that leverages concepts from erasure coding to mitigate the effect of stragglers (slow nodes) in distributed computation clusters, especially for matrix computation problems. In this work, we present a class of distributed matrix-vector multiplication schemes that are based on codes in the Rosenbloom-Tsfasman metric and universally decodable matrices. Our schemes take into account the inherent computation order within a worker node. In particular, they allow us to effectively leverage partial computations performed by stragglers (a feature that many prior works lack). An additional main contribution of our work is a companion matrix-based embedding of these codes that allows us to obtain sparse and numerically stable schemes for the problem at hand. Experimental results confirm the effectiveness of our techniques.
\end{abstract}

\section{Introduction}
Distributed computation clusters are routinely used in domains such as machine learning and scientific computing. 
In these applications, datasets are often so large that they cannot be housed in the disk of a single server. Furthermore, processing the data on a single server is either infeasible or unacceptably slow. Thus, the data and the processing is distributed and processed across a large number of nodes.

While large clusters have numerous advantages, they also present newer operational challenges.
These clusters (which can be heterogeneous in nature) suffer from the problem of ``stragglers" which are defined as slow nodes (node failures are an extreme form of a straggler). It is evident that the overall speed of a computation on these clusters is typically dominated by stragglers in the absence of a sophisticated assignment of tasks to the worker nodes.


In recent years, approaches based on coding theory (referred to as ``coded computation") have been effectively used for straggler mitigation
\cite{yu2017polynomial, tang2018bound, lee2017high, lee2018speeding, duttaCG16, mallick2018rateless, wangLS18, kia2018exploit, das2018c}. Coded computation offers significant benefits for specific classes of problems, e.g., matrix computations. We illustrate this by means of a matrix-vector multiplication example in Fig. \ref{fig:Q_example}, where a matrix $\bfA$ is block-row decomposed as $\bfA^T = [\bfA_0^T ~\bfA_1^T~\bfA_2^T]^T$. Each worker node is given the responsibility of computing two submatrix-vector products so that the computational load on each worker is $2/3$-rd of the original. It can be observed that even if one worker fails, there is enough information for a master node to compute the final result. However, this requires the master node to solve simple systems of equations. This approach can be generalized (and also adapted for matrix multiplication) by using Reed-Solomon (RS) code like approaches \cite{yu2017polynomial, tang2018bound, lee2017high, lee2018speeding, duttaCG16}. These methods allow the master node to recover $\bfA \bfx$ if any $\tau$ of the worker nodes complete their computation; $\tau$ is called the recovery threshold.


A significant amount of prior work treats stragglers as node failures (see \cite{kia2018exploit,das2018c,mallick2018rateless} for exceptions), or, equivalently from the point of view of coding theory, as erasures. This matches the conventional erasure coding problem very well and allows the adaptation of well-known approaches, e.g, RS codes to the problem of distributed matrix computations. However, there are certain features of the distributed matrix-vector multiplication problem that distinguish it from classical erasure correction that we now discuss.
\begin{itemize}
	\item {\it Leveraging partial computation performed by stragglers.} Each worker node operates in a sequential fashion on its assigned rows, e.g., in Fig. \ref{fig:Q_example}, worker $W_0$, first computes $\hat{\bfA}_{00}\bfx$ and only then $\hat{\bfA}_{01}\bfx$. If node $0$ is a straggler (but not a failure), ignoring the partial computation it performs will be wasteful.
	\item {\it Numerically stable decoding.} The RS-based approach requires the master node to solve a real Vandermonde system of linear equations or equivalently perform polynomial interpolation. It is well recognized that real Vandermonde matrices have a rather large condition number\footnote{While there is literature on choosing good evaluation points to reduce the condition number, in the distributed matrix vector multiplication context, we require decoding from any $\tau$ evaluation points. This makes the worst case condition number quite bad.} which translates into significant numerical issues in recovering $\bfA \bfx$. This numerical issue is especially important in Krylov subspace methods for solving large linear systems of equations \cite{simoncini2007recent} (which repeatedly compute matrix-vector products) and in machine learning,  where gradient computations are often approximate. 

	
	\item {\it Dealing with sparse $\bfA$ matrices.} The case when the matrix $\bfA$ is sparse is often an important one in practice. RS-based approaches typically generate submatrices that are sent to the worker nodes by combining a large number of rows of $\bfA$, thus destroying the inherent sparsity of the problem. This can significantly increase the computation time \cite{wangLS18} at the worker nodes. Thus, techniques that only require sparse combinations of the rows of $\bfA$ are of great interest.
\end{itemize}
\subsection{Main contributions of our work}
We present a class of distributed matrix-vector multiplication schemes that provably leverage partial computations by stragglers, while possessing a numerically stable decoding algorithm. These schemes are related to codes in the Rosenbloom-Tsfasman metric \cite{rosenbloomT97} and universally decodable matrices (UDM) \cite{ganesan2007existence} that were presented in different contexts. Roughly speaking, while the RS-based approach corresponds to polynomial evaluation/interpolation, our approach can be viewed as working with polynomials with roots of higher multiplicity. An additional main contribution of our work is the usage of companion matrices \cite{ward94matrix} that allow for an embedding of finite-field matrices into the real field; this significantly improves the condition numbers of the relevant matrices. 

\section{Problem Formulation}
\label{sec:prob_formu}



We consider a scenario where the master node has a matrix $\bfA$ and vector $\bfx$ (both real-valued) and is connected to $N$ worker nodes. For convenience, for arbitrary positive integer $n$, let $[n] \triangleq \{0, \dots, n-1\}$. The master node first partitions $\bfA$ into $\Delta$ block-rows (or submatrices) denoted by $\bfA_0, \cdots, \bfA_{\Delta - 1}$,  each of the same dimension. Following this, it generates submatrices denoted $\hat{\bfA}_{i,j}$, $i \in [N], j \in [\ell]$ (of the same dimension as the $\bfA_i$'s) such that worker $W_i$ is sent submatrices $\hat{\bfA}_{i,j}, j \in [\ell]$, and the vector $\bfx$. Let $\gamma \triangleq \ell/\Delta$. Then, each worker is assigned the equivalent of a $\gamma$-fraction of the rows of $\bfA$. In this paper, we assume $\bfA$ is large enough so that $\Delta$ can be chosen large enough. Throughout this paper the submatrices $\hat{\bfA}_{i,j}$ will be linear combinations of $\bfA_0, \dots, \bfA_{\Delta - 1}$, such that the master node only calculates scalar multiples and sums of $\bfA_i$'s. 


In what follows, we say that worker $W_i$ has processed a submatrix $\hat{\bfA}_{i,j}$ if it has calculated $\hat{\bfA}_{i,j} \bfx$.
A key feature of the distributed matrix-vector multiplication problem is that the matrices $\hat{\bfA}_{i,j}$ are processed sequentially in the order $\hat{\bfA}_{i,0}, \hat{\bfA}_{i,1}, \dots$, i.e.,
a worker node $W_i$ processes $\hat{\bfA}_{i,j}$ only if it has finished processing $\hat{\bfA}_{i,0}, \dots, \hat{\bfA}_{i,j-1}$. Each time a worker node computes a product or a block of consecutive products, it sends the result to the master node. Our system requirement dictates that the master node should be able to decode $\bfA \bfx$ as long as it receives a minimum number of products from the worker nodes. Fig. \ref{fig:Q_example} demonstrates a system we consider.

\begin{figure}[t]
	\centering
	\includegraphics[width=90mm]{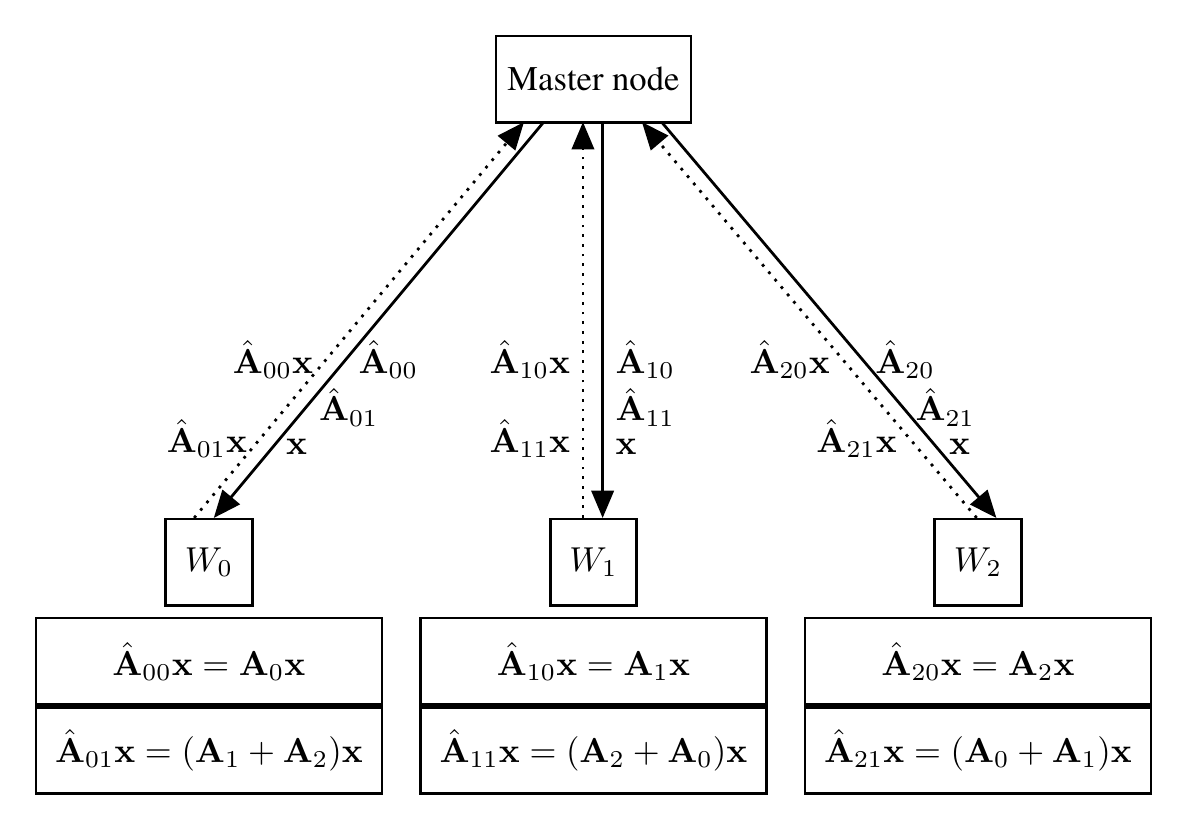}
	\caption{ {\small Matrix $\bfA$ is partitioned into three submatrices $\bfA_0, \bfA_1$ and $\bfA_3$ by rows. The master node transmits $\bfx$ and the $\hat{\bfA}_{ij}$ submatrices to the workers (shown by solid lines). The workers transmit the submatrix vector products back to the master node (shown by dotted lines). It can be verified that the master node can decode if any three submatrices are processed, while respecting the sequential computation order (from top to bottom).}}
	\label{fig:Q_example}
\end{figure}

%

It is evident that the properties of a given scheme depend upon the properties of the matrices $\hat{\bfA}_{i,j}$, $i \in [N], j \in [\ell]$. To specify this encoding we discuss constructions of collections of matrices that have certain desired properties. Some of our constructions are first designed over a finite field and then embedded into $\mathbb{R}$ using an appropriately defined procedure. Accordingly, we define certain rank conditions that depend on an underlying field of operation denoted by $\mathbb{K}$. We will explicitly specify $\mathbb{K}$ when discussing the constructions. Consider $N$ matrices $G_k$, $0\le k < N$ over $\mathbb{K}$ with dimension $\Delta \times \ell$ and let $G_k(i,j)$ represent the $(i,j)$-th entry of $G_k$. Let $\bfv = (v_0, \cdots, v_{N-1})$. We define the set 
\begin{align*}
	\Psi_{N,\ell,s}^{=Q_{\mathrm{b}}} = \left \{\bfv~\bigg{|}~ v_i \in [\ell], i\in [N], \sum_{i\in [N]} \left\lfloor \frac{v_i}{s} \right\rfloor =Q_{\mathrm{b}} \right \}.
\end{align*}
\begin{definition}{\it $s$-weak full-rank matrices.}
	\label{defn:full_rank_cond}
	Let $N, \ell, \Delta, s$ be positive integers such that $s$ divides $\ell$  and $\Delta$. Let $Q_{\mathrm{b}} = \Delta/s$.  Consider matrices $G_k, k \in [N]$, of dimension $\Delta \times \ell$. Let $\bfv \in \Psi_{N,\ell,s}^{=Q_{\mathrm{b}}}$ and let $\Delta_\bfv = \sum_{i \in [N]} v_i \geq Q_{\mathrm{b}} s$. If the $\Delta \times \Delta_\bfv$ matrix $\bfG$ composed of the first $v_0$ columns of $G_0$, the first $v_1$ columns of $G_1$, $\dots$, and the first $v_{N-1}$ columns of $G_{N-1}$, has full rank over $\mathbb{K}$, i.e., $\text{rank}_{\mathbb{K}}(\bfG)=\Delta$ for all $\bfv \in \Psi_{N,\ell,s}^{=Q_{\mathrm{b}}}$ we say the collection $\{G_k\}_{k=0}^{N-1}$ satisfies the $s$-weak full-rank condition. 
\end{definition}

The collection $\{G_k\}_{k=0}^{N-1}$ is used to obtain the $\ell$ submatrices stored in worker $W_k$ when $\mathbb{K} = \mathbb{R}$ as 
\begin{align*}
	\hat{\bfA}_{k,j} &= \sum_{i \in [\Delta]} G_k(i,j) \bfA_{i}, ~~\text{for~}k \in [N],j \in [\ell].
\end{align*}

Consider first the case $s=1$ and assume that worker $W_i$ has finished processing $v_i$ submatrices and $v_0+\cdots + v_{N-1} \geq \Delta$. Let $\bfG$ be as specified in the definition above. 
It is not too hard to see that the system requirement of decoding from any $Q_{\mathrm{b}} = \Delta$ submatrix-vector products is equivalent to the condition that $\bfG$ is full-rank over $\mathbb{R}$ for all possible patterns $(v_0, \cdots, v_{N-1})$. Thus, designing $\{G_k\}_{k=0}^{N-1}$ that satisfy Definition \ref{defn:full_rank_cond} is sufficient for the problem at hand.

Values of $s > 1$ correspond to  a relaxation of this condition. Specifically, suppose that each worker node returns the results in blocks of size $s$. For instance, worker node $W_i$ computes $\hat{\bfA}_{i,0}\bfx, \hat{\bfA}_{i,1} \bfx , \dots, \hat{\bfA}_{i,s-1}\bfx$ and then reports the result back to the master. Following this, it focuses on the next block $\hat{\bfA}_{i,s} \bfx, \dots, \hat{\bfA}_{i,2s-1} \bfx$, and so on. In this case, decoding by the master node is guaranteed if it receives any $Q_{\mathrm{b}}$ blocks of size $s$ (this explains our choice of subscript $\mathrm{b}$ in $Q_{\mathrm{b}}$). If the $s$-weak full-rank condition holds for $s=1$, we will refer to the system as satisfying the \emph{strong full-rank condition}. 

\begin{remark}
	When $s=1$, then the master node can recover $\bfA \bfx$ when any $Q_b$ submatrices have been processed across the $N$ workers, i.e., the worst case computational load on the system, measured at the granularity of a submatrix is $\Delta$. If $s > 1$, then the worst case computational load can be as high as
	\begin{align}
		\label{eq:comp_load}
		\Delta \left(1 + \frac{(N-1)(s-1)}{\Delta} \right).
	\end{align} For our constructions, the second term $\frac{(N-1)(s-1)}{\Delta}$ can be made as small as desired by choosing a large enough $\Delta$.
\end{remark}

%
\begin{example}
	Consider the system in Fig. \ref{fig:Q_example} with $N=3$, $\Delta = 3$, $\ell = 2$. Matrix $\bfA$ is partitioned into three submatrices by rows, $\bfA_0, \bfA_1, \bfA_2$. Each worker node is assigned two submatrices and the vector $\bfx$. The following real-valued matrices satisfy the conditions in Definition \ref{defn:full_rank_cond} for $s=1$ (see Fig. \ref{fig:Q_example} for the corresponding $\hat{\bfA}_{i,j}$ matrices).
	\begin{align*}
		G_0=\begin{bmatrix}
			1&0\\
			0&1\\
			0&1
		\end{bmatrix},~
		G_1=\begin{bmatrix}
			0&1\\
			1&0\\
			0&1
		\end{bmatrix}, \text{~and~}
		G_2=\begin{bmatrix}
			0&1\\
			0&1\\
			1&0
		\end{bmatrix}.
	\end{align*}
\end{example}

\section{Coded schemes satisfying \\the strong full rank condition}
\label{sec:strong_sec}
In this section, we present two schemes that satisfy the strong full-rank condition. The first scheme is essentially an embedding of an RS code in the matrix-vector multiplication framework and has appeared in \cite{lee2018speeding}. The second one is inspired by the constructions in \cite{rosenbloomT97, ganesan2007existence}.

Let $u(x) = \sum_{k=0}^d u_k x^k$ be a polynomial of degree $d$ with real coefficients, i.e., $u(x) \in \mathbb{R}[x]$ where $\mathbb{R}[x]$ denotes the ring of polynomials with real coefficients. Let $u^{(j)}(x)$ denote the $j$-th derivative of $u(x)$. It is evident that
\begin{align}
	u^{(j)}(x) = \sum_{k=0}^d u_k \binom{k}{j} j!~ x^{k-j}, \label{eq:real_poly_der}
\end{align}
where $\binom{k}{i} = 0$ if $k < i$.
Furthermore, note that we can also represent $u(x)$ by considering its Taylor series expansion around a point $\beta \in \mathbb{R}$, i.e.,
\begin{align}
	u(x) = \sum_{k=0}^d \frac{u^{(k)}(\beta)}{k!} (x - \beta)^k. \label{eq:real_Taylor_series}
\end{align}
It is well known that $u(x)$ has a zero of multiplicity $m$ at $\beta \in \mathbb{R}$ if and only if $u^{(i)}(\beta) = 0$ for $0 \leq i < m$ and $u^{(m)}(\beta) \neq 0$.

\subsection{RS-based scheme}
\label{sec:MDS}In the first scheme we simply choose the columns of $G_k$ for $k \in [N]$ to correspond to a polynomial of degree $\Delta -1$ being evaluated at distinct points in $\mathbb{R}$, i.e.,
\begin{align}
	\label{eq:MDS_Gi}
	G_k(i,j) = \beta_{k,j}^i, \text{~for~} i \in [\Delta], j \in [\ell].
\end{align}
where $\beta_{k,j} \in \mathbb{R}$ are distinct for $k \in [N], j \in [\ell]$. 

\subsection{UDM-based scheme}
\label{sec:UDM}
Our second construction works by choosing the columns of $G_k$ corresponding to the evaluations of a polynomial and its derivatives of order $1, \dots, \ell -1$. We first choose $N$ distinct real numbers $\beta_0, \dots, \beta_{N-1}$. For worker node $k$, we choose the $j$-th column in correspondence with the evaluation of the $j$-th derivative of a degree-$(\Delta-1)$ polynomial at value $\beta_k$ scaled by $j!$ ({\it cf.} Eq. (\ref{eq:real_poly_der})). Thus, for $k \in [N], i \in [\Delta]$ and $j \in [\ell]$, 
\begin{align}
	G_k(i,j) = \begin{cases}
		\binom{i}{j} \beta_k^{i-j} & \text{if $i \geq j$,}\\
		0 & \text{otherwise.}
	\end{cases}
	\label{eq:RT_assign}
\end{align}
We note here that there is another choice of matrix, denoted $G_{*}$ that can be used instead of the above choices for one of the workers.
For  $i \in [\Delta]$ and $j \in [\ell]$, we let
\begin{align}
G_{*}(i,j) = \begin{cases}
1, & \text{if $i = N-1-j$,}\\
0 & \text{otherwise.}
\end{cases}
	\label{eq:RT_anti_assign}
\end{align}
\subsection{Properties of the Coded Schemes}

\begin{claim}
	The $N$ matrices defined in Section \ref{sec:MDS} and Section \ref{sec:UDM} satisfy the strong full-rank condition in Definition \ref{defn:full_rank_cond}.
\end{claim}
\begin{proof}
	Consider any vector pattern $\bfv=(v_0, \ldots, v_{N-1})$ such that $v_0+\cdots + v_{N-1} = \Delta$. Let $\bfG$ be composed of the first $v_k$ columns of $G_k$, $k=0,\ldots, N-1$.
	For the RS-based construction in (\ref{eq:MDS_Gi}), it is evident that $\bfG$ is a Vandermonde matrix.
	As the $\beta_{k,j}$'s are distinct, $\bfG$ has full rank.
	For the UDM-based scheme, if all workers are chosen based on (\ref{eq:RT_assign}), the result follows from the determinant of a generalized Vandermonde determinant \cite{sobczyk2002}. On the other hand, assume without loss of generality that the $(N-1)$-th worker is assigned the $G_*$ matrix ({\it cf.} Eq. (\ref{eq:RT_anti_assign})). In this case, $\bfG$ can be written as
	\begin{align*}
	\bfG=\begin{bmatrix}
	\bfB_1 &0\\
	\bfB_2 &\bfB_3
	\end{bmatrix},
	\end{align*}
	where $\bfB_3$ is a $v_{N-1}\times v_{N-1}$ matrix with ones on the anti-diagonal and $\bfB = [\bfB_1^T~\bfB_2^T]^T$ are composed of the first $v_k$ columns of $G_k$, $k=0, \cdots, N-2$. Once again, the generalized Vandermonde determinant formula \cite{sobczyk2002} shows that $\bfB_1$ is full rank. This coupled with the fact that $\bfB_3$ is also full rank, gives us the required result.
\end{proof}


It is evident that the above constructions satisfy the strong full-rank condition. However, experimental results (see also \cite{gautschi1990stable}) show that these constructions result in badly conditioned $\bfG$ matrices in the worst case. In addition, both (\ref{eq:MDS_Gi}) and (\ref{eq:RT_assign}) result in dense linear combinations of $\bfA$, rendering them unsuitable in the scenario when $\bfA$ is sparse. Nevertheless, the UDM-based construction (\ref{eq:RT_assign}), provides a systematic way to take into account the sequential processing order of the worker nodes.
\begin{remark}
	The RS-based scheme is in one-to-one correspondence with polynomial interpolation from any $\Delta$ (out of $N \ell$) distinct evaluation points. The UDM-based scheme uses much fewer evaluation points (only $N$) but is equivalent to interpolating a polynomial with roots of higher multiplicity.
\end{remark}

\section{Coded schemes satisfying \\ the weak full rank condition}
\label{sec:weak_full_rank}

Our second class of constructions produces schemes that satisfy the $s$-weak full-rank condition. However, they have excellent numerical stability and are much sparser than those discussed in Section \ref{sec:strong_sec}. These schemes are obtained by first constructing a collection of matrices over a finite field $\bbF_{p^n}$ (where $p$ is prime) and then embedding the finite field matrices into real field by companion matrix.
Towards this end, let $\tilde{u}(x) = \sum_{k=0}^d \tilde{u}_k x^k$ be a polynomial with coefficients from $\bbF_{p^n}$, i.e., $\tilde{u}(x) \in \bbF_{p^n}[x]$. The $i$-th Hasse derivative\footnote{To avoid confusion with the case of real-valued polynomials, we superscript the finite field polynomials with ~\textasciitilde~ and represent the Hasse derivatives with square brackets. } of $\tilde{u}(x)$ is defined as
\begin{align}
	\label{eq:Hesse_UDM}
	\tilde{u}^{[i]}(x) = \sum_{k=0}^d \binom{k}{i} \tilde{u}_k x^{k-i},
\end{align}
where we emphasize that the quantity $\binom{k}{i}$ is interpreted as a element of $\bbF_p$. In this scenario, it can be shown that $\tilde{u}(x)$ has a zero of multiplicity $m$ at a point $\beta \in \bbF_{p^n}$ (or in an appropriate extension field) if $\tilde{u}^{[i]}(\beta) = 0$ for $0 \leq i < m$ and $\tilde{u}^{[m]}(\beta) \neq 0$.

The work of \cite{rosenbloomT97,ganesan2007existence} shows that the following matrices $G_k, k \in [N]$, satisfy the strong full-rank condition over $\mathbb{K} = \bbF_{p^n}$, assuming $p^n\ge N + 1$.
\begin{align}
	G_k(i,j) = \begin{cases}
		\binom{i}{j} \beta_k^{i-j} & \text{(if $i \geq j$)}\\
		0 & \text{otherwise.}
	\end{cases}
	\label{eq:RT_ff_assign}
\end{align}
where $\beta_k, k \in [N]$ are distinct non-zero elements in $\bbF_{p^n}$. We remark here that while the expression above is the same as the one in (\ref{eq:RT_assign}), the elements of (\ref{eq:RT_ff_assign}) lie in $\bbF_{p^n}$.

One reason for considering the matrices in (\ref{eq:RT_ff_assign}) is as follows. Suppose that we operate over $\bbF_{2^n}$, i.e., $p=2$. Note that the calculation in (\ref{eq:RT_ff_assign}) is equivalent to computing $\binom{i}{j}$ over the integers and reducing it modulo $2$. In particular, this implies that whenever $\binom{i}{j}$ is even, the corresponding matrix entry will be zero. Thus, over finite fields, the $\{G_k\}_{k=0}^{N-1}$ matrices obtained using (\ref{eq:RT_ff_assign}) are likely much sparser than those obtained from (\ref{eq:RT_assign}).

\begin{example}
	\label{ex:UDM_nonbin}
	Let $p = 2, n = 3, \ell = 3, \Delta = 4, N = 6$. Consider the polynomial $\tilde{u}(x) = \tilde{u}_0 + \tilde{u}_1x + \tilde{u}_2x^2 + \tilde{u}_3x^3$ over $\bbF_{8}$. Its $i$-th Hasse derivatives, $i = 0, 1, 2$ are
	\begin{align*}
		\tilde{u}^{[0]}(x) &= \tilde{u}_0 + \tilde{u}_1x + \tilde{u}_2x^2 + \tilde{u}_3x^3,\\
		\tilde{u}^{[1]}(x) &= \tilde{u}_1 + \tilde{u}_3x^2, \text{~and}\\
		\tilde{u}^{[2]}(x) & = \tilde{u}_2 + \tilde{u}_3x.
	\end{align*}
	Note here that $\tilde{u}^{[1]}(x)$ has only two non-zero coefficients, whereas when considering derivatives over the reals, it will have three non-zero coefficients.
	Then,
	\begin{align}
		\label{eq:HasseEx}
		G_i&=\begin{bmatrix}
			1 & 0 & 0\\
			\beta_i & 1 & 0\\
			\beta_i^2 & 0 & 1\\
			\beta_i^3 & \beta_i^2 & \beta_i
		\end{bmatrix},
	\end{align}
	where $\beta_i \in \bbF_8$ and $\beta_i$ values are distinct for $i \in [N]$.
\end{example}

A natural question arises if it is possible to somehow ``embed" the matrices defined in (\ref{eq:RT_ff_assign}) into corresponding real matrices such that the conditions of Definition \ref{defn:full_rank_cond} hold (for real matrices).
This does not appear to be a straightforward problem.  For example, simply requiring distinct $\beta_k$'s is not sufficient. For instance, if we choose $\beta_0 = 1$ and $\beta_1 = -1$ then the matrix
$$
	\bfG=\begin{bmatrix}
	1 & 0 & 1 & 0\\
	1 & 1 & -1 & 1 \\
	1 & 0 & 1 & 0 \\
	1 & 1 & -1 & 1
    \end{bmatrix}
$$
obtained by choosing the first two columns of $G_0$ and the first two columns of $G_1$ is singular.

\begin{remark}
If the $\beta_i, i \in [N]$ are chosen randomly from a large enough subset of $\mathbb{R}$, then we can assert that the collection will satisfy the strong full-rank property with high probability. To see this, let $\beta_i \in \bbF_{p^n}, i \in [N]$ be indeterminates for now and consider $\bfG$ for any pattern $(v_0, \cdots, v_{N-1})$. The determinant of $\bfG$ is a multivariate polynomial $\tilde{\Lambda}(\beta_0, \dots, \beta_{N-1})$ with coefficients from $\bbF_p$. The results of \cite{rosenbloomT97,ganesan2007existence} certainly imply that $\tilde{\Lambda}(\beta_0, \dots, \beta_{N-1})$ is not identically zero. Now, consider the determinant (polynomial) of $\bfG$ denoted $\Lambda(\beta_0, \dots, \beta_{N-1})$  obtained by considering $\beta_i \in \mathbb{R}, i \in [N]$, i.e., $\Lambda(\beta_0, \dots, \beta_{N-1})$ has integer coefficients. Clearly, $\tilde{\Lambda}(\beta_0, \dots, \beta_{N-1})$ can be obtained by reducing each coefficient of $\Lambda(\beta_0, \dots, \beta_{N-1})$ modulo $p$. Therefore, $\Lambda(\beta_0, \dots, \beta_{N-1})$ is also not identically zero. It follows that the product of all the real multivariate polynomials corresponding to the relevant $\bfG$'s is not identically zero. The result then follows, by choosing a large enough subset of the reals and applying the Schwartz-Zippel lemma.
\end{remark}

Next, we utilize a representation of $\bbF_{p^n}$ by $n \times n$ matrices over $\bbF_p$ \cite{ward94matrix}.
Let $\bbF_p[x]$ denote the ring of polynomials in $x$ with coefficients from $\bbF_p$. 
Let $\alpha$ be a primitive element in $\mathbb{F}_{p^n}$ and let $\pi_{\alpha}(x) = x^n + \sum_{i=0}^{n-1}\pi_i x^i \in \bbF_p[x]$ denote the primitive polynomial associated with $\alpha$. The $n\times n$ companion matrix (over $\bbF_p$) associated with $\pi_{\alpha}(x)$ is
\begin{align}
	\label{eq:companion}
	C = \begin{bmatrix}
		0 & 0 & \cdots & 0 & -\pi_0\\
		1 & 0 & \cdots & 0 & -\pi_1\\
		0 & 1 & \cdots & 0 & -\pi_2\\
		\vdots & \vdots & \ddots & \vdots & \vdots\\
		0 & 0 & \vdots & 1 & -\pi_{n-1}
	\end{bmatrix}.
\end{align}

Define $\frakC(p, n)=\{0, I, C, C^2,\cdots, C^{p^{n}-2}\}$ with matrix addition and multiplication over $\bbF_p$, where $0$ denotes $n\times n$ zero matrix and $I$ denotes $n\times n$ identity matrix. Then it is well-known \cite{ward94matrix} that the $\frakC(p, n)$  forms a finite field of size $p^n$ and is therefore isomorphic to $\bbF_{p^n}$. In particular, the mapping $\zeta(\alpha^l)=C^l$, $\zeta(0)=0$, maps the elements in $\bbF_{p^n}$ to their corresponding matrix representation. In this work, we need another isomorphism. The elements of $\bbF_{p^n}$ are represented by polynomials in $\alpha$ of degree smaller than $n$ with regular polynomial addition and multiplication being reduced to lower powers by using $\pi_\alpha(\alpha) = 0$. Let $\Gamma: \bbF_{p^n} \rightarrow \bbF^n_{p}$ represent the mapping of a polynomial $a(\alpha)$ to its vector representation. The addition of $a_1(\alpha)$ and $a_2(\alpha)$ is mapped to $\Gamma(a_1)+\Gamma(a_2)$. The product of $a_1(\alpha)$ and $a_2(\alpha)$ is mapped to $a_1(C) \Gamma(a_2)$.

To see that this is a valid isomorphism, we have the following argument that establishes the equivalence of multiplication with $\alpha$ in $\bbF_{p^n}$ and left multiplication by $C$.
Let $b_0+b_1\alpha + \cdots + b_{n-1}\alpha^{n-1}$ be an element of $\bbF_{p^n}$. Then
\begin{align*}
&\alpha(b_0+b_1\alpha+\cdots + b_{n-1}\alpha^{n-1})\\
=&b_0\alpha+b_1\alpha^2+\cdots+b_{n-2}\alpha^{n-1}-b_{n-1}(\sum_{i=0}^{n-1}\pi_i\alpha^i)\\
=&-\pi_0b_{n-1}+(b_0-\pi_1b_{n-1})\alpha+(b_1-\pi_2b_{n-1})\alpha^2\\
&+\cdots+(b_{n-2}-\pi_{n-1}b_{n-1})\alpha^{n-1}.
\end{align*}
It can be seen that $C[b_0 ~b_1~ \cdots~ b_{m-1}]^T$ gives the same result. The isomorphism of $\bbF_{p^n}$  and $\frakC(p, n)$ shows that each element of $\bbF_{p^n}$ can be represented as a power of $C$. The result is then obtained by inductively applying the equivalence presented above.

\begin{lemma}
	\label{lemma:bin_prj_field}
	Let $\bfB$ be a $n \times n$ matrix with entries from $\bbF_{p^m}$. Let $\tilde{\bfB}$ denote the $mn \times mn$ matrix obtained by applying the map $\zeta$ to each entry of $\bfB$. Note that $\det(\bfB) \in \bbF_{p^m}$ and $\det(\tilde{\bfB}) \in \bbF_p$. We claim that
	\begin{align}
		\det(\bfB) \neq 0 \Longleftrightarrow \det(\tilde{\bfB}) \neq 0. \label{eq:det_eq}
	\end{align}
	Furthermore, let $\hat{\bfB}$ denote the $mn \times mn$ matrix over the integers $\mathbb{Z}$ obtained by mapping each element of $\tilde{\bfB}$ to the corresponding integer in $\{0, \dots, p-1\}$.  If $\det(\bfB) \neq 0$ we have $\det(\hat{\bfB}) \neq 0$ over the reals.
\end{lemma}

\begin{proof}
		Suppose that $\det(\bfB) \neq 0$ but $\det(\tilde{\bfB}) = 0$. Note that this implies that there exists a non-zero vector $\tilde{\bfy} = [\tilde{y}_1^T~\tilde{y}_2^T~\dots~\tilde{y}_n^T]^T \in \bbF_{p}^{mn}$ where $\tilde{y}_i \in \bbF^m_p$ such that
		\begin{align}
		\tilde{\bfB} \tilde{\bfy}  = 0. \label{eq:eq1}
		\end{align}
		Now we use the isomorphism presented above. Let $\bfy = [{y}_1~{y}_2~\dots~{y}_n]^T \in \bbF_{p^m}^n$ be obtained by applying $\Gamma^{-1}$ to $\tilde{\bfy}$. Therefore, relation (\ref{eq:eq1}), equivalently implies that
		\begin{align}
		\bfB {\bfy}  = 0, \label{eq:eq2}
		\end{align}
		where the above equation is understood to be over $\bbF_{p^m}$. However, this is a contradiction since ${\bfy} \neq 0$ and $\det(\bfB) \neq 0$. The reverse conclusion can be obtained in a similar manner.
	Note that $\det(\tilde{\bfB}) \in \bbF_p$. It can also be equivalently computed by finding $\det(\hat{\bfB})$ over reals and reducing the result modulo $p$. Thus, we have that $\det(\hat{\bfB}) \neq 0$ over reals.
\end{proof}

We now present the construction of systems that satisfy the $s$-weak full-rank property.

\begin{table*}[t]
	\centering
	\caption{Performance comparison for system with $N=6, \gamma = 3/4$ and $Q_{\mathrm{b}} = 4$.}
	\label{table:perf}
	\begin{tabular}{|c||c|c|c|c|c|c|}
		\hline
		Scheme&$\Delta$&$\ell$ & s & Max. Cond. Num. & Avg. Cond. Num. & Density of $G_k$\\
		\hline
		RS based scheme& 4 & 3 & &$ 5.1\times 10^3$ & $ 334$ & $100\%$\\
		\hline
		RS $+$ Companion Matrix of $\mathrm{GF}(2^5)$& 20 & 15 & 5 & $3.4\times 10^4$ & $814$ & $51\%$\\
		\hline
		RS $+$ Embedding from $\mathrm{GF}(19)$ & 4 & 3 & & $ 7.3\times 10^3$ & $312$ & $100\%$\\
		\hline
		RS $+$ Companion Matrix of $\mathrm{GF}(3^3)$& 12 & 9 & 3 & $1.5\times 10^3$ & $98$ & $71\%$\\
		\hline
		UDM-based scheme& 4 & 3 & & $ 6.1\times 10^3$ & $265$ & $75\%$\\
		\hline
		UDM $+$ Embedding from $\mathrm{GF}(7)$& 4 & 3& & $ 1.5\times 10^3$ & $98$ & $75\%$\\
		\hline
		UDM $+$ Companion Matrix of $\mathrm{GF}(2^3)$& 12 & 9 & 3 & $583$ & $99$ & $32\%$\\
		\hline
		UDM $+$ Companion Matrix of $\mathrm{GF}(3^2)$& 8 & 6 & 2 & $182$ & $23$& $36\%$\\
		\hline
	\end{tabular}
\end{table*}

\begin{table*}[t]
	\centering
	\caption{Performance comparison of different extension fields for a system with $N=15, \gamma = 1/2$ and $Q_{\mathrm{b}} = 4$.}
	\label{table:perf2}
	\begin{tabular}{|c||c|c|c|c|c|c|c|}
		\hline
		Scheme&$\Delta$&$\ell$ & s & Max. Cond. Num. & Avg. Cond. Num. & Density of $G_k$\\
		\hline
		RS $+$ Companion Matrix $\mathrm{GF}(2^5)$ & 20 & 10 & 5 & $2.8\times 10^5$& $751$ & $53\%$\\
		\hline
		RS $+$ Companion Matrix $\mathrm{GF}(5^3)$ & 12 & 6 & 3 & $1.1\times 10^5$& $183$ & $83\%$\\
		\hline
		RS $+$ Companion Matrix $\mathrm{GF}(3^4)$ & 16 & 8 & 4 & $3.5\times 10^4$& $202$ & $67\%$\\
		\hline
		UDM $+$ Companion Matrix $\mathrm{GF}(2^4)$& 16 & 8 & 4 & $3.7\times 10^4$ & $286$ & $33\%$\\
		\hline
		UDM $+$ Companion Matrix $\mathrm{GF}(5^2)$& 8 & 4 & 2 & $1.1\times 10^4$ & $86$ & $62\%$\\
		\hline
		UDM $+$ Companion Matrix $\mathrm{GF}(3^3)$& 12 & 6 & 3 & $624$ & $96$ & $41\%$\\
		\hline
	\end{tabular}
\end{table*}

\begin{lemma}
	\label{lemma:bin_proj_code}
	Let $G_k$, $0\le k<N$, be a collection of $N$ matrices with size $\Delta \times \ell$ over $\bbF_{p^n}$ that satisfy  the strong full-rank property. Consider the $N$ matrices $G'_k$ of dimension $n\Delta  \times n\ell$ over $\bbF_p$, where $G'_k$ is obtained by applying the mapping $\zeta$ to each entry of $G_k$. Then, the collection $\{G'_k\}_{k=0}^{N-1}$ satisfies the $n$-weak full rank condition over $\mathbb{R}$.
\end{lemma}
\begin{proof}
	This is an immediate consequence of Lemma \ref{lemma:bin_prj_field}.
\end{proof}

\begin{example}
	Consider collection of $N = 6$ matrices presented in Example \ref{ex:UDM_nonbin} over $\bbF_{8}$. Let the primitive polynomial over $\bbF_{8}$ be $\pi_\alpha(x) = x^3+x+1$. 
	Suppose that $\beta_k = \alpha^k, k \in [N]$. Then,
	\begin{align*}
		G_k=\begin{bmatrix}
			1 & 0 & 0\\
			\alpha^k & 1 & 0\\
			\alpha^{2k} & 0 & 1\\
			\alpha^{3k} & \alpha^{2k} & \alpha^k
		\end{bmatrix},
		&&
		G'_k=\begin{bmatrix}
			I & 0 & 0\\
			C^k & I & 0\\
			C^{2k} & 0 & I\\
			C^{3k} & C^{2k} & C^k
		\end{bmatrix}.
	\end{align*}
	By Lemma \ref{lemma:bin_proj_code}, the collection $\{G'_k\}_{k=0}^{5}$ satisfies the $3$-weak full rank condition over $\mathbb{R}$.
	
	%
\end{example}
We note here that Lemma \ref{lemma:bin_proj_code} can also be applied to an RS code defined over a finite field.

\begin{remark}
	Our proposed scheme requires us to operate over an extension field large enough so that $p^n \geq N+1$ for the UDM based approach and $p^n \geq N\ell + 1$ for the RS-based approach. Thus, the second term in the worst case computational load ({\it cf.} Eq. (\ref{eq:comp_load})) can be made as small as desired by choosing $\Delta$ large enough. 
	Increasing $\Delta$ does come at the cost of high condition numbers ({\it cf.} Section \ref{sec:results}).
\end{remark}

\section{Comparisons of the different schemes}
\label{sec:results}

In this section, we compare the performance of the different schemes that have been proposed in this work. For each scheme, we construct all possible matrices $\bfG$ based on $\{G_k\}_{k=0}^{N-1}$ and $\Psi_{N,\ell,s}^{=Q_{\mathrm{b}}}$ and calculate their condition number.  We report the maximum and average condition number of all such possible $\bfG$'s. Furthermore, we also report the average number of non-zero elements in the $G_k$ matrices for each collection.


In Table \ref{table:perf} we report results for a system with $N = 6$ workers and storage capacity for each worker $\gamma = 3/4$. 
For the ``RS-based scheme", we set $\beta_{i,j}$, $i = 0, \ldots, 5$, $j = 0,1,2$, in  (\ref{eq:MDS_Gi}) to 18 equally spaced reals within the interval $[-1, 1]$. For the ``UDM-based scheme", we set $\beta_k$ in (\ref{eq:RT_assign}), $k = 0, \ldots, 5$ to
6 equally spaced reals within the interval $[-1, 1]$.
For ``RS + Embedding from $\mathrm{GF}(19)$", we construct (\ref{eq:MDS_Gi}) over $\mathrm{GF}(19)$. Note that the field size is the least prime number that is greater or equal to the number of evaluation points. Then we embed (\ref{eq:MDS_Gi}) into $\mathbb{R}$ by using the natural mapping of $\mathrm{GF}(19)$ into the integers. We construct ``UDM + Embedding from $\mathrm{GF}(7)$" in a similar manner.
It can be seen that the condition number of ``UDM scheme + Embedding  from $GF(7)$" is the lowest when compared the other three schemes discussed thus far.

The other rows of Table \ref{table:perf} correspond to the companion matrix approach. In each of these cases we first design the RS-based or the UDM-based scheme over the corresponding extension field and then use the companion matrix idea introduced in Section \ref{sec:weak_full_rank}.
One can observe that the RS + Companion matrix schemes typically have high condition number. This is because the size of the companion matrix needs to be large enough to accommodate $N \ell$ evaluation points. The UDM + Companion matrix schemes can work with extension fields larger than $N$, so their companion matrices tend to be smaller. Another advantage of the companion matrix approach is that the schemes are much sparser. Indeed, the ``UDM + Companion matrix $\mathrm{GF}(3^2)$" in Table \ref{table:perf} not only has a very low worst case condition number but also a sparsity level of 36\% which is the second lowest among all the schemes.

To better understand the performance corresponding to different choices of extension field, we consider a larger system with $N = 15$, $\gamma = 1/2$ in Table \ref{table:perf2}. It can be observed that the RS-based scheme is worse than the UDM-based scheme. Another observation is that the ``UDM + Companion matrix $\mathrm{GF}(3^3)$" has the lowest condition number and the $G_k$ matrices become sparser when the size of the companion matrix increases.

\end{document}